\documentclass[12pt]{article}

\usepackage[margin=1in]{geometry}
\usepackage{CJKutf8, amsmath, amssymb, amsthm, authblk, hyperref}
\usepackage[basic]{complexity}

\newtheorem{lemma}{Lemma}
\newtheorem{theorem}{Theorem}
\newtheorem{corollary}{Corollary}

\theoremstyle{definition}
\newtheorem{assumption}{Assumption}

\DeclareMathOperator{\tr}{tr}
\DeclareMathOperator{\dist}{dist}

\usepackage[style=trad-unsrt, citestyle=numeric-comp, giveninits=true, doi=false, url=false, eprint=false, isbn=false]{biblatex}

\begin{document}

\title{High-precision simulation of finite-size thermalizing systems at long times}

\author{Yichen Huang (黄溢辰)\thanks{yichenhuang@fas.harvard.edu}}
\affil{Department of Physics, Harvard University, Cambridge, Massachusetts 02138, USA}

\begin{CJK}{UTF8}{gbsn}

\maketitle

\end{CJK}

\begin{abstract}

To simulate thermalizing systems at long times, the most straightforward approach is to calculate the thermal properties at the corresponding energy. In a quantum many-body system of size $N$, for local observables and many initial states, this approach has an error of $O(1/N)$, which is reminiscent of the finite-size error of the equivalence of ensembles. In this paper, we propose a simple and efficient numerical method so that the simulation error is of higher order in $1/N$. This finite-size error scaling is proved by assuming the eigenstate thermalization hypothesis.

\end{abstract}

\section{Introduction}

Thermalization is one of the most remarkable phenomena in nature. Consider an isolated quantum many-body system governed by the Hamiltonian $H$ and initialized in a pure state $|\psi(0)\rangle$. Let $S$ be a small subsystem and $\bar S$ be its complement (rest of the system). Let
\begin{equation}
|\psi(t)\rangle=e^{-iHt}|\psi(0)\rangle,\quad\psi(t)_S:=\tr_{\bar S}|\psi(t)\rangle\langle\psi(t)|
\end{equation}
be the state and its reduced density matrix of $S$ at time $t\in\mathbb R$. Let
\begin{equation} \label{eq:th}
\rho=e^{-\beta H}/\tr(e^{-\beta H})
\end{equation}
be a thermal state at the same energy, i.e., the inverse temperature $\beta$ is determined from
\begin{equation} \label{eq:it}
E_\psi:=\langle \psi(0)|H|\psi(0)\rangle=\tr(\rho H).
\end{equation}
Thermalization means \cite{GE16} that at long times,
\begin{itemize}
\item $\psi(t)_S$ becomes almost independent of time, i.e., its temporal fluctuation vanishes in the thermodynamic limit. This is known as equilibration.
\item $\|\psi^\infty_S-\rho_S\|_1$ vanishes in the thermodynamic limit, where $\psi^\infty_S$ is the equilibrated $\psi(t)_S$, and $\rho_S:=\tr_{\bar S}\rho$ is the reduced density matrix of the thermal state (\ref{eq:th}).
\end{itemize}
Despite the existence of counterexamples, it is generally believed that thermalization usually occurs.

Our goal is to compute $\psi^\infty_S$ in thermalizing systems. To this end, the baseline is to compute $\rho_S$. For one-dimensional systems, this can be done very efficiently on a (classical) computer \cite{Ara69, KAA21, BCP22}. While it is (numerically) exact in the thermodynamic limit, this baseline has a finite-size error: In a system of size $N$, for many initial states (e.g., states with exponential decay of correlations),
\begin{equation} \label{eq:s}
\|\psi^\infty_S-\rho_S\|_1=O(1/N).
\end{equation}
Note that the prefactor hidden in the big-O notation is generically nonzero (although in fine-tuned cases, $\|\psi^\infty_S-\rho_S\|_1$ can be of higher order in $1/N$).

The scaling (\ref{eq:s}) can be understood as the finite-size error of the equivalence of ensembles (here we only explain the intuition; a quantitative argument is given later in this paper). Let $\rho^\textnormal{mc}$ be the uniform classical mixture of eigenstates in a small energy window such that
\begin{equation}
\tr(\rho^\textnormal{mc}H)=\tr(\rho H),
\end{equation}
i.e., $\rho^\textnormal{mc}$ represents the microcanonical ensemble at the same energy. Let $\rho^\textnormal{mc}_S:=\tr_{\bar S}\rho^\textnormal{mc}$. It is well known that
\begin{equation} \label{eq:cm}
\|\rho^\textnormal{mc}_S-\rho_S\|_1\sim1/N.
\end{equation}
Let $\{|j\rangle\}_{j=1,2,\ldots}$ be a complete set of eigenstates of $H$. We write the initial state as a superposition of eigenstates
\begin{equation}
    |\psi(0)\rangle=\sum_jc_j|j\rangle.
\end{equation}
Under mild assumptions (e.g., the spectrum of $H$ has non-degenerate gaps), equilibration and dephasing can be proved \cite{Rei08, LPSW09} so that $\psi^\infty_S=\tr_{\bar S}\psi^\infty$, where
\begin{equation} \label{eq:pp}
\psi^\infty:=\sum_jp'_j|j\rangle\langle j|,\quad p'_j:=|c_j|^2
\end{equation}
is the so-called diagonal ensemble \cite{RDO08}. While the microcanonical and diagonal ensembles are different, it is unsurprising that their finite-size error scaling with respect to the canonical ensemble is the same.

For a moderate system size, say, $N=50$ spins, Eq.~(\ref{eq:s}) predicts a finite-size error of $\simeq0.02$, which may not be small enough for practical purposes. Unfortunately, Eq.~(\ref{eq:s}), as a physical scaling relation, cannot be improved. In this paper, we computationally overcome the limitation suggested by Eq.~(\ref{eq:s}). We propose a simple and efficient numerical method that constructs an approximation to $\psi^\infty_S$ from $\rho$ (rather than directly outputs $\rho_S$) such that the approximation error is of higher order in $1/N$.

The eigenstate thermalization hypothesis (ETH) \cite{Deu91, Sre94, RDO08, DKPR16, Deu18} provides an explanation for the emergence of statistical mechanics from the unitary evolution of quantum systems. Although not all thermalizing systems satisfy the ETH \cite{HH22}, almost all of them do. The ETH says that the eigenstate expectation values are a smooth function of energy density. It implies the equivalence of ensembles with the finite-size error scaling (\ref{eq:s}), (\ref{eq:cm}) \cite{Sre99}. We show that in the thermodynamic limit, not only the scaling of $\|\psi^\infty_S-\rho_S\|_1$ but also the leading term of $\psi^\infty_S-\rho_S$ can be calculated using the ETH. Adding this term to $\rho_S$, the error is reduced to higher order in $1/N$.

\section{Results}

Consider a chain of $N$ spins so that the dimension of the Hilbert space is $d=d_\textnormal{loc}^N$, where $d_\textnormal{loc}$ (a constant) is the local dimension of each spin. The system is governed by a translation-invariant local Hamiltonian $H$ with periodic boundary conditions. (We impose translational invariance and periodic boundary conditions for simplicity. Neither is absolutely necessary for our argument below.) Let $\mathbb T$ be the (unitary) lattice translation operator, which acts on the computational basis states as
\begin{equation}
\mathbb T(|x_1\rangle\otimes|x_2\rangle\otimes\cdots\otimes|x_N\rangle)=|x_2\rangle\otimes|x_3\rangle\otimes\cdots\otimes|x_N\rangle\otimes|x_1\rangle
\end{equation}
with $x_l\in\{0,1,\ldots,d_\textnormal{loc}-1\}$ for $l=1,2,\ldots,N$. We write the Hamiltonian as
\begin{equation} \label{eq:H}
H=\sum_{l=0}^{N-1}H_l,\quad H_l=\mathbb T^{-l}h\mathbb T^l,
\end{equation}
where $h$ is a Hermitian operator acting on spins at positions $1,2,\ldots,k$ for some constant $k$. Assume without loss of generality that $\tr h=0$ (traceless) and $\|h\|=1$ (unit operator norm). Let $\{|j\rangle\}_{j=1}^d$ be a complete set of translation-invariant eigenstates of $H$ with corresponding energies $\{E_j\}$.

Suppose that the initial state $|\psi(0)\rangle$ has exponential decay of correlations. This includes all product states (each spin is disentangled from all other spins), whose correlation length is zero. Note that $|\psi(0)\rangle$ may not be translationally invariant.

If we measure the energy density of $|\psi(0)\rangle$, we obtain $E_j/N$ with probability $p'_j$. Exponential decay of correlations implies that this probability distribution is concentrated around the energy density $e_\psi:=E_\psi/N$ of $|\psi(0)\rangle$ \cite{Ans16}. Thus, only a neighborhood of $e_\psi$ is relevant. We assume the ETH in such a neighborhood. 

\begin{assumption} [eigenstate thermalization hypothesis] \label{ass}
Let $\epsilon$ be an arbitrarily small positive constant. For any local operator $A$ with $\|A\|=O(1)$, there is a sequence of functions $\{f_N:[e_\psi-\epsilon,e_\psi+\epsilon]\to\{z\in\mathbb C:|z|=O(1)\}\}$ (one for each system size $N$) such that
\begin{equation} \label{eq:ass}
\big|\langle j|A|j\rangle-f_N(E_j/N)\big|\le1/\poly(N)
\end{equation}
for all $j$ with $|E_j/N-e_\psi|\le\epsilon$, where $\poly(N)$ denotes a polynomial of sufficiently high degree in $N$. We assume that each $f_N(x)$ is smooth in the sense of having a Taylor expansion to some low order around $x=e_\psi$:
\begin{equation} \label{eq:t}
f_N(e_\psi+\delta)=f_N(e_\psi)+f'_N(e_\psi)\delta+f''_N(e_\psi)\delta^2/2+f'''_N(e_\psi)\delta^3/6+O(\delta^4),\quad\forall|\delta|\le\epsilon.
\end{equation}
\end{assumption}

In quantum chaotic systems, it was proposed analytically \cite{Sre99} and supported by numerical simulations \cite{KIH14} that the left-hand side of (\ref{eq:ass}) is exponentially small in $N$. For our purposes, however, a (much weaker) inverse polynomial upper bound suffices.

We are ready to present our main result. Suppose $E_\psi$ is not too close to the edges of the spectrum such that $\beta$, determined from Eq.~(\ref{eq:it}), is $O(1)$. For notational simplicity, let
\begin{equation}
G:=H-E_\psi,\quad g:=h-e_\psi.
\end{equation}

\begin{theorem} \label{thm}
For any local operator $A$ with $\|A\|=O(1)$, Assumption \ref{ass} implies that
\begin{multline} \label{eq:fd}
\tr(\psi^\infty A)=\tr(\rho A)+O(1/N^2)\\
+\frac1{2N\tr(\rho Gg)}\left(1-\frac{\langle\psi(0)|G^2|\psi(0)\rangle}{N\tr(\rho Gg)}\right)\left(\frac{\tr(\rho G^2g)\tr(\rho GA)}{\tr(\rho Gg)}+\tr\big(\rho G^2\big(\tr(\rho A)-A\big)\big)\right).
\end{multline}
\end{theorem}

The second line of Eq.~(\ref{eq:fd}) is $O(1/N)$ and can be calculated from the thermal properties of the system at energy density $e_\beta$.

In the language of reduced density matrices, Eq.~(\ref{eq:fd}) can be stated as\footnote{I thank Soonwon Choi for pointing out that Eq.~(\ref{eq:fd}) can be stated as Eq.~(\ref{eq:cor}).}
\begin{multline} \label{eq:cor}
\bigg\|\frac1{2N\tr(\rho Gg)}\left(1-\frac{\langle\psi(0)|G^2|\psi(0)\rangle}{N\tr(\rho Gg)}\right)\left(\frac{\tr(\rho G^2g)\tr_{\bar S}(\rho G)}{\tr(\rho Gg)}+\tr(\rho G^2)\rho_S-\tr_{\bar S}(\rho G^2)\right)\\
+\rho_S-\psi_S^\infty\bigg\|_1=O(1/N^2).
\end{multline}
Thus, we obtain an approximation to $\psi_S^\infty$ such that the finite-size error is $O(1/N^2)$.

\section{Proofs}

If we measure the energy density of either $\psi^\infty$ or $\rho$, we obtain a probability distribution concentrated around $e_\psi$. Thus, it suffices to study $f_N(e_\psi),f'_N(e_\psi),f''_N(e_\psi),\ldots$, which characterize $f(x)$ near $x=e_\psi$. Our proof consists of two steps. First, Lemma \ref{l:8} expresses $f_N(e_\psi),f'_N(e_\psi),f''_N(e_\psi)$ with $\rho$, whose properties can be computed efficiently. Then, we use this information to calculate $\tr(\psi^\infty A)$. Both steps make heavy use of the Taylor expansion (\ref{eq:t}). Similar methods were used in Refs.~\cite{Sre96, Sre99, DKPR16, Hua22AP, Hua22ATMP}. Our main technical contribution is a rigorous calculation of the finite-size error, especially in the finite-temperature case. Both our intermediate and final results (Theorem \ref{thm}, Lemma \ref{l:8}, Corollary \ref{eigs}) are new.

The support of an operator is the set of spins it acts non-trivially on. Let $\dist(A_1,A_2)$ be the distance between the supports of two operators $A_1,A_2$. In one-dimensional translation-invariant systems, the thermal state $\rho$ at any inverse temperature $\beta=O(1)$ has exponential decay of correlations \cite{Ara69, BCP22}
\begin{equation} \label{eq:edc}
|\tr(\rho A_1A_2)-\tr(\rho A_1)\tr(\rho A_2)|=\|A_1\|\|A_2\|O(e^{-\dist(A_1,A_2)/\xi}),
\end{equation}
where the correlation length $\xi$ is a constant that depends on $\beta$.

Let
\begin{equation}
p_j:=e^{-\beta E_j}\big/\sum_{j=1}^de^{-\beta E_j},\quad F_j:=E_j-E_\psi,\quad G_l:=H_l-e_\psi
\end{equation}
so that $F_1,F_2,\ldots,F_d$ are the eigenvalues of $G$.

\begin{lemma} [moments] \label{l:moment}
\begin{gather}
\sum_{j=1}^dp_jF_j=\tr(\rho G)=0,\label{eq:1}\\
\sum_{j=1}^dp_jF_j^2=\tr(\rho G^2)=N\tr(\rho Gg)=\Theta(N),\label{eq:2}\\
\sum_{j=1}^dp_jF_j^3=\tr(\rho G^3)=N\tr(\rho G^2g)=O(N),\label{eq:3}\\
\sum_{j=1}^dp_jF_j^4=\tr(\rho G^4)=3\tr^2(\rho G^2)+O(N)=3N^2\tr^2(\rho Gg)+O(N),\label{eq:4}\\
\sum_{j=1}^dp_jF_j^5=\tr(\rho G^5)=O(N^2),\label{eq:5}\\
\sum_{j=1}^dp_jF_j^6=\tr(\rho G^6)=O(N^3).\label{eq:6}
\end{gather}
\end{lemma}

\begin{proof} [Proof of Eq.~(\ref{eq:2})]
The second equality is due to the translational invariance of $G$ and $\rho$. $-\tr(\rho G^2)$ is the heat capacity with respect to the inverse temperature $\beta$. Exponential decay of correlations (\ref{eq:edc}) implies that it is extensive.
\end{proof}

\begin{proof} [Proof of Eqs.~(\ref{eq:3}), (\ref{eq:5}), (\ref{eq:6})]
Equation (\ref{eq:6}) is a special case of Lemma 4.1 in Ref.~\cite{Ans16}. To prove Eq. (\ref{eq:3}), we improve the proof of Lemma 4.1 in Ref.~\cite{Ans16}. Let
\begin{equation}
D(l_1,l_2):=\min\{|l_1-l_2|,N-|l_1-l_2|\}
\end{equation}
be the distance between $G_{l_1}$ and $G_{l_2}$. For any tuple $(l_1,l_2,\ldots,l_n)$, let
\begin{equation}
D(l_1,l_2,\ldots,l_n):=\max_{i\in\{1,2,\ldots,n\}}\min_{j\neq i}D(l_i,l_j)
\end{equation}
so that
\begin{equation}
|\{(l_1,l_2,l_3)\in\{0,1,\ldots,N-1\}^3:D(l_1,l_2,l_3)=r\}|=NO(r+1).
\end{equation}
Since $\tr(\rho G_j)=0$ for all $j$, exponential decay of correlations (\ref{eq:edc}) implies that
\begin{equation}
\tr(\rho G_{l_1}G_{l_2}G_{l_3})=O(e^{-D(l_1,l_2,l_3)/\xi}).
\end{equation}
Therefore,
\begin{multline}
\tr(\rho G^3)=\sum_{l_1,l_2,l_3=0}^{N-1}\tr(\rho G_{l_1}G_{l_2}G_{l_3})=\sum_{r\ge0}\sum_{l_1,l_2,l_3:D(l_1,l_2,l_3)=r}\tr(\rho G_{l_1}G_{l_2}G_{l_3})\\
\le\sum_{r\ge0}NO(r+1)e^{-r/\xi}=O(N).
\end{multline}
Equation (\ref{eq:5}) can be proved in the same way as Eq.~(\ref{eq:3}). 
\end{proof}

\begin{proof} [Proof of Eq.~(\ref{eq:4})]
Let
\begin{gather}
    T_r:=\{(l_1,l_2,l_3,l_4)\in\{0,1,\ldots,N-1\}^4:D(l_1,l_2,l_3,l_4)=r\},\\
    T'_{r,1}:=\big\{(l_1,l_2,l_3,l_4)\in\{0,1,\ldots,N-1\}^4:\max\{D(l_1,l_2),D(l_3,l_4)\}=r\big\},\\
    T'_{r,2}:=\big\{(l_1,l_2,l_3,l_4)\in\{0,1,\ldots,N-1\}^4:\max\{D(l_1,l_3),D(l_2,l_4)\}=r\big\},\\
    T'_{r,3}:=\big\{(l_1,l_2,l_3,l_4)\in\{0,1,\ldots,N-1\}^4:\max\{D(l_1,l_4),D(l_2,l_3)\}=r\big\},\\
    T_{r,i}=T_r\cap T'_{r,i},\quad i=1,2,3
\end{gather}
so that
\begin{equation}
T_r=T_{r,1}\cup T_{r,2}\cup T_{r,3},\quad|T_{r,1}\cap T_{r,2}|=NO(r+1),\quad|T'_{r,i}\setminus T_{r,i}|=NO(r^2).
\end{equation}
Exponential decay of correlations (\ref{eq:edc}) implies that
\begin{equation}
\tr(\rho G_{l_1}G_{l_2}G_{l_3}G_{l_4})=O(e^{-D(l_1,l_2,l_3,l_4)/\xi}),\quad\tr(\rho G_{l_1}G_{l_2})=O(e^{-D(l_1,l_2)/\xi}).
\end{equation}
Therefore,
\begin{multline}
\tr(\rho G^4)=\sum_{l_1,l_2,l_3,l_4=0}^{N-1}\tr(\rho G_{l_1}G_{l_2}G_{l_3}G_{l_4})=\sum_{r\ge0}\sum_{(l_1,l_2,l_3,l_4)\in T_r}\tr(\rho G_{l_1}G_{l_2}G_{l_3}G_{l_4})\\
\approx\sum_{i=1}^3\sum_{r\ge0}\sum_{(l_1,l_2,l_3,l_4)\in T_{r,i}}\tr(\rho G_{l_1}G_{l_2}G_{l_3}G_{l_4}),
\end{multline}
where the approximation error is upper bounded by
\begin{equation}
\sum_{r\ge0}NO(r+1)e^{-r/\xi}=O(N).
\end{equation}
Since $\min\{\tr(\rho G_{l_1}G_{l_2}),\tr(\rho G_{l_3}G_{l_4})\}=O(e^{-r/\xi})$ for $(l_1,l_2,l_3,l_4)\in T'_{r,1}$,
\begin{multline}
\tr^2(\rho G^2)=\sum_{r\ge0}\sum_{(l_1,l_2,l_3,l_4)\in T'_{r,1}}\tr(\rho G_{l_1}G_{l_2})\tr(\rho G_{l_3}G_{l_4})\\
\approx\sum_{r\ge0}\sum_{(l_1,l_2,l_3,l_4)\in T_{r,1}}\tr(\rho G_{l_1}G_{l_2})\tr(\rho G_{l_3}G_{l_4}),
\end{multline}
where the approximation error is upper bounded by
\begin{equation}
\sum_{r\ge0}NO(r^2)e^{-r/\xi}=O(N).
\end{equation}
Since
\begin{align}
&\sum_{\substack{(l_1,l_2,l_3,l_4)\in T_{r,1}\\\dist(G_{l_1}G_{l_2},G_{l_3}G_{l_4})\le r}}|\tr(\rho G_{l_1}G_{l_2}G_{l_3}G_{l_4})-\tr(\rho G_{l_1}G_{l_2})\tr(\rho G_{l_3}G_{l_4})|=NO(r^2+1)e^{-r/\xi},\\
&\sum_{\substack{(l_1,l_2,l_3,l_4)\in T_{r,1}\\\dist(G_{l_1}G_{l_2},G_{l_3}G_{l_4})>r}}|\tr(\rho G_{l_1}G_{l_2}G_{l_3}G_{l_4})-\tr(\rho G_{l_1}G_{l_2})\tr(\rho G_{l_3}G_{l_4})|\nonumber\\
&=\sum_{r'>r}\sum_{\substack{(l_1,l_2,l_3,l_4)\in T_{r,1}\\\dist(G_{l_1}G_{l_2},G_{l_3}G_{l_4})=r'}}O(e^{-r'/\xi})=\sum_{r'>r}NO(r+1)e^{-r'/\xi}=NO(r+1)e^{-r/\xi},
\end{align}
we obtain
\begin{equation}
\sum_{r\ge0}\sum_{(l_1,l_2,l_3,l_4)\in T_{r,1}}|\tr(\rho G_{l_1}G_{l_2}G_{l_3}G_{l_4})-\tr(\rho G_{l_1}G_{l_2})\tr(\rho G_{l_3}G_{l_4})|=O(N).
\end{equation}
Similarly,
\begin{equation}
\sum_{r\ge0}\sum_{(l_1,l_2,l_3,l_4)\in T_{r,2}}|\tr(\rho G_{l_1}G_{l_2}G_{l_3}G_{l_4})-\tr(\rho G_{l_1}G_{l_3})\tr(\rho G_{l_2}G_{l_4})|=O(N).
\end{equation}
We complete the proof by combining the equations above.
\end{proof}

If we measure the energy of the thermal state $\rho$, the measurement results are concentrated.
\begin{lemma} [\cite{Ans16}] \label{Mar}
For any $\epsilon>0$,
\begin{equation}
\sum_{j:|F_j|\ge N\epsilon}p_j=O(e^{-\Omega(\epsilon\sqrt N)}).
\end{equation}
\end{lemma}

This lemma allows us to upper bound the total contribution of all eigenstates away from $E_\psi$. Let $C=O(1)$ be a sufficiently large constant such that
\begin{equation} \label{tail}
\sum_{j:|F_j|\ge\Lambda}p_j|F_j|^m\le q,\quad\Lambda:=C\sqrt{N}\log N,\quad q:=1/\poly(N)
\end{equation}
for $m=0,1,2$, where $\poly(N)$ denotes a polynomial of sufficiently high degree in $N$.

For notational simplicity, let $x\overset{\delta}=y$ denote $|x-y|\le\delta$.

\begin{lemma} \label{l:8}
For any local operator $A$ with $\|A\|=O(1)$, Assumption \ref{ass} implies that
\begin{gather}
f_N(e_\psi)=\tr(\rho A)+\frac{\tr(\rho G^2g)\tr(\rho GA)}{2N\tr^2(\rho Gg)}+\frac{\tr(\rho G^2(\tr(\rho A)-A))}{2N\tr(\rho Gg)}+O(1/N^2),\label{eq:b1}\\
f_N'(e_\psi)=\tr(\rho GA)/\tr(\rho Gg)+O(1/N),\label{eq:b2}\\
f_N''(e_\psi)=\frac{\tr(\rho G^2(A-\tr(\rho A)))}{\tr^2(\rho Gg)}-\frac{\tr(\rho G^2g)\tr(\rho GA)}{\tr^3(\rho Gg)}+O(1/N).
\end{gather}
\end{lemma}

\begin{proof}
We have
\begin{align} \label{eq:27}
&\tr(\rho A)=\sum_{j=1}^dp_j\langle j|A|j\rangle\overset{O(q)}=\sum_{j:|F_j|<\Lambda}p_j\langle j|A|j\rangle\overset{1/\poly(N)}=\sum_{j:|F_j|<\Lambda}p_jf_N(E_j/N)\nonumber\\
&=\sum_{j:|F_j|<\Lambda}p_j\left(f_N(e_\psi)+\frac{f_N'(e_\psi)F_j}N+\frac{f_N''(e_\psi)F_j^2}{2N^2}+\frac{f_N'''(e_\psi)F_j^3}{6N^3}+O(F_j^4/N^4)\right)\nonumber\\
&\overset{O(q)}=\sum_{j=1}^dp_j\left(f_N(e_\psi)+\frac{f_N'(e_\psi)F_j}N+\frac{f_N''(e_\psi)F_j^2}{2N^2}+\frac{f_N'''(e_\psi)F_j^3}{6N^3}+O(F_j^4/N^4)\right)\nonumber\\
&=f_N(e_\psi)+\frac{f_N''(e_\psi)\tr(\rho Gg)}{2N}+\frac{f_N'''(e_\psi)\tr(\rho G^2g)}{6N^2}+O(1/N^2)\nonumber\\
&=f_N(e_\psi)+\frac{f_N''(e_\psi)\tr(\rho Gg)}{2N}+O(1/N^2),
\end{align}
where we used inequality (\ref{tail}), the ETH (\ref{eq:ass}), and the Taylor expansion
\begin{equation} \label{eq:tayloreth}
f_N(E_j/N)=f_N(e_\psi)+f_N'(e_\psi)F_j/N+f_N''(e_\psi)F_j^2/(2N^2)+f_N'''(e_\psi)F_j^3/(6N^3)+O(F_j^4/N^4)
\end{equation}
in the steps marked with ``$O(q)$,'' ``$1/\poly(N)$,'' and from the first to the second line, respectively. Similarly,
\begin{align} \label{derivc}
&\tr(\rho GA)=\sum_{j=1}^dp_jF_j\langle j|A|j\rangle\overset{O(q)}=\sum_{j:|F_j|<\Lambda}p_jF_j\langle j|A|j\rangle\overset{1/\poly(N)}=\sum_{j:|F_j|<\Lambda}p_jF_jf_N(E_j/N)\nonumber\\
&=\sum_{j:|F_j|<\Lambda}p_j\left(f_N(e_\psi)F_j+\frac{f_N'(e_\psi)F_j^2}N+\frac{f_N''(e_\psi)F_j^3}{2N^2}+O(F_j^4/N^3)\right)\nonumber\\
&\overset{O(q)}=\sum_{j=1}^dp_j\left(f_N(e_\psi)F_j+\frac{f_N'(e_\psi)F_j^2}N+\frac{f_N''(e_\psi)F_j^3}{2N^2}+O(F_j^4/N^3)\right)\nonumber\\
&=f'(e_\psi)\tr(\rho Gg)+\frac{f''(e_\psi)\tr(\rho G^2g)}{2N}+O(1/N)=f'(e_\psi)\tr(\rho Gg)+O(1/N).
\end{align}
Thus, we obtain Eq. (\ref{eq:b2}). Furthermore,
\begin{align} \label{2derivc}
&\tr(\rho G^2A)=\sum_{j=1}^dp_jF_j^2\langle j|A|j\rangle\overset{O(q)}=\sum_{j:|F_j|<\Lambda}p_jF_j^2\langle j|A|j\rangle\overset{1/\poly(N)}=\sum_{j:|F_j|<\Lambda}p_jF_j^2f_N(E_j/N)\nonumber\\
&=\sum_{j:|F_j|<\Lambda}p_j\left(f_N(e_\psi)F_j^2+\frac{f_N'(e_\psi)F_j^3}N+\frac{f_N''(e_\psi)F_j^4}{2N^2}+\frac{f_N'''(e_\psi)F_j^5}{6N^3}+O(F_j^6/N^4)\right)\nonumber\\
&\overset{O(q)}=\sum_{j=1}^dp_j\left(f_N(e_\psi)F_j^2+\frac{f_N'(e_\psi)F_j^3}N+\frac{f_N''(e_\psi)F_j^4}{2N^2}+\frac{f_N'''(e_\psi)F_j^5}{6N^3}+O(F_j^6/N^4)\right)\nonumber\\
&=f_N(e_\psi)\tr(\rho G^2)+f_N'(e_\psi)\tr(\rho G^2g)+\frac{f_N''(e_\psi)\tr(\rho G^4)}{2N^2}+\frac{f_N'''(e_\psi)\tr(\rho G^5)}{6N^3}+O(1/N)\nonumber\\
&=f_N(e_\psi)N\tr(\rho Gg)+\frac{\tr(\rho G^2g)\tr(\rho GA)}{\tr(\rho Gg)}+\frac{3f_N''(e_\psi)\tr^2(\rho Gg)}{2}+O(1/N).
\end{align}
We complete the proof of the lemma by solving (\ref{eq:27}) and (\ref{2derivc}).
\end{proof}

Equation (\ref{eq:b1}) shows the difference between the eigenstate and thermal expectation values.

\begin{corollary} \label{eigs}
Let $|j\rangle$ be an eigenstate whose energy $E_j$ is (very close to) zero. For any traceless local operator $A$ with $\|A\|=1$, Assumption \ref{ass} implies that
\begin{equation} \label{eq:eigs}
\langle j|A|j\rangle=\frac{\tr(H^2h)\tr(HA)-\tr(Hh)\tr(H^2A)}{2N\tr^2(Hh)}+O(1/N^2).
\end{equation}
\end{corollary}

There always exists \cite{Hua22ATMP} a traceless local operator $A$ such that the coefficient of the $1/N$ term on the right-hand side of Eq.~(\ref{eq:eigs}) is non-zero.

In the language of reduced density matrices, Eq.~(\ref{eq:b1}) can be stated as\footnote{Again I thank Soonwon Choi for pointing out that Eq.~(\ref{eq:b1}) can be stated as Eq.~(\ref{eq:eigd}).}
\begin{equation} \label{eq:eigd}
\left\|\rho_S+\frac{\tr(\rho G^2g)\tr_{\bar S}(\rho G)}{2N\tr^2(\rho Gg)}+\frac{\tr(\rho G^2)\rho_S-\tr_{\bar S}(\rho G^2)}{2N\tr(\rho Gg)}-\tr_{\bar S}|j\rangle\langle j|\right\|_1=O(1/N^2),
\end{equation}
where $\rho$ is the thermal state whose energy is $E_j$.

\begin{lemma}
Recall the definition (\ref{eq:pp}) of $p'_j$.
\begin{gather}
\sum_{j=1}^dp'_jF_j=\langle\psi(0)|G|\psi(0)\rangle=0,\quad\sum_{j=1}^dp'_jF_j^2=\langle\psi(0)|G^2|\psi(0)\rangle=O(N),\\
\sum_{j=1}^dp'_jF_j^3=\langle\psi(0)|G^3|\psi(0)\rangle=O(N),\quad\sum_{j=1}^dp'_jF_j^4=\langle\psi(0)|G^4|\psi(0)\rangle=O(N^2).
\end{gather}
\end{lemma}

\begin{proof}
Let $G'_l:=H_l-\langle\psi(0)|H_l|\psi(0)\rangle$ so that
\begin{equation} \label{eq:g}
G=\sum_{l=0}^{N-1}G'_l,\quad\langle\psi(0)|G'_l|\psi(0)\rangle=0.
\end{equation}
Since $|\psi(0)\rangle$ has exponential decay of correlations, the lemma can be proved by expanding $G^2,G^3,G^4$ using Eq.~(\ref{eq:g}).
\end{proof}

Since $|\psi(0)\rangle$ has exponential decay of correlations, Lemma \ref{Mar} and Eq.~(\ref{tail}) remain valid upon replacing $p_j$ by $p'_j$.

\begin{proof} [Proof of Theorem \ref{thm}]
\begin{align}
&\tr(\psi^\infty A)=\sum_{j=1}^dp'_j\langle j|A|j\rangle\overset{O(q)}=\sum_{j:|F_j|<\Lambda}p'_j\langle j|A|j\rangle\overset{1/\poly(N)}=\sum_{j:|F_j|<\Lambda}p'_jf_N(E_j/N)\nonumber\\
&=\sum_{j:|F_j|<\Lambda}p'_j\left(f_N(e_\psi)+\frac{f_N'(e_\psi)F_j}N+\frac{f_N''(e_\psi)F_j^2}{2N^2}+\frac{f_N'''(e_\psi)F_j^3}{6N^3}+O(F_j^4/N^4)\right)\nonumber\\
&\overset{O(q)}=\sum_{j=1}^dp'_j\left(f_N(e_\psi)+\frac{f_N'(e_\psi)F_j}N+\frac{f_N''(e_\psi)F_j^2}{2N^2}+\frac{f_N'''(e_\psi)F_j^3}{6N^3}+O(F_j^4/N^4)\right)\nonumber\\
&=f_N(e_\psi)+\frac{f_N''(e_\psi)\langle\psi(0)|G^2|\psi(0)\rangle}{2N^2}+O(1/N^2).
\end{align}
We complete the proof by using Lemma \ref{l:8}.
\end{proof}

\section*{Notes}

Recently, I became aware of a related work \cite{CWX+24}. It studies a different problem, but there is some overlap between the intermediate technical aspects of their work and mine.

\section*{Acknowledgments}

I would like to thank Fernando G.S.L. Brand\~ao for pointing out the use of exponential decay of correlations in the finite-temperature case while collaborating on a related project \cite{HBZ19}; Yi Tan and Norman Y. Yao for discussions and collaboration on related work; Nicole Yunger Halpern for pointing out important references; Daniel K. Mark for discussions; Soonwon Choi for suggestions on improving the presentation and for his contribution acknowledged in the footnotes.

\printbibliography

\end{document}